\documentclass[11pt]{article}
\usepackage{fullpage}
\usepackage{natbib}

\usepackage{amsmath}
\usepackage{amsthm}
\usepackage{amsfonts}
\usepackage{graphicx}

\usepackage[colorlinks=true, allcolors=blue]{hyperref}

\theoremstyle{plain}
\newtheorem{theorem}{Theorem}
\newtheorem{corollary}[theorem]{Corollary}
\newtheorem{lemma}[theorem]{Lemma}

\newcommand{\ud}{\,\mathrm{d}}
\renewcommand{\epsilon}{\varepsilon}

\newcommand{\bb}{\mathbf{b}}
\newcommand{\x}{\mathbf{x}}
\newcommand{\p}{\mathbf{p}}

\allowdisplaybreaks

\title{\bf Randomized learning-augmented auctions\\
with revenue guarantees}
\author{
    Ioannis Caragiannis \and
    Georgios Kalantzis}
\date{Department of Computer Science, Aarhus University\\
    iannis@cs.au.dk, 202202869@post.au.dk}

\begin{document}
\maketitle

\begin{abstract}
We consider the fundamental problem of designing a truthful single-item auction with the challenging objective of extracting a large fraction of the highest agent valuation as revenue. Following a recent trend in algorithm design, we assume that the agent valuations belong to a known interval, and a (possibly erroneous) prediction for the highest valuation is available. Then, auction design aims for high consistency and robustness, meaning that, for appropriate pairs of values $\gamma$ and $\rho$, the extracted revenue should be at least a $\gamma$- or $\rho$-fraction of the highest valuation when the prediction is correct for the input instance or not. We characterize all pairs of parameters $\gamma$ and $\rho$ so that a randomized $\gamma$-consistent and $\rho$-robust auction exists. Furthermore, for the setting in which robustness can be a function of the prediction error, we give sufficient and necessary conditions for the existence of robust auctions and present randomized auctions that extract a revenue that is only a polylogarithmic (in terms of the prediction error) factor away from the highest agent valuation.
\end{abstract}

\section{Introduction}
We revisit one of the most central topics in microeconomic theory, namely the design of {\em single-item auctions}, under the lens of the recent trend of {\em learning-augmented algorithms}~\citep{MV20}. The basic auction setting consists of a seller who has an item for sale and potential buyers (or agents), each having private valuations for the item. An auction is then used to sell the item. The agents submit their bids to the seller, who then decides who should get the item and at which price. 

Agents are strategic; they may consider bidding non-truthfully (i.e., submit a bid that is different than their valuation for the item) if this can increase their utility. Since non-truthful bids may make reasoning about the agents' behaviour and the auction's outcome difficult, a large part of auction design theory has focused on truthful auctions.

For single-item auctions (and, single-parameter environments, more generally), the auction designer has essentially to define two components: an allocation rule and a payment rule. Both are functions taking the agent bids as input. The allocation rule returns the agent who should get the item (if any), and the payment rule returns the payment the agents should give to the seller.

The design of truthful auctions requires careful selection of both the allocation and the payment rules. For example, the celebrated second-price auction~\citep{V61} allocates the item to the agent with the highest bid, who pays the second-highest bid to the seller. The second-price auction is truthful. Even though proving this formally is rather easy and usually serves as warming-up material in introductory textbooks to mechanism design (e.g., see~\cite{R16}), doing so for more sophisticated auctions can be tricky. Fortunately, the beautiful theory of~\cite{M81} provides us with a powerful toolset for single-item auction design (and, more generally, for mechanism design in single-parameter environments). The infamous Myerson's Lemma restricts the allocation rules that can be used in truthful auctions to those having a monotonicity property. Furthermore, the allocation rule identifies the payment function in an almost unique way.

The second-price auction has additional nice economic properties, e.g., it maximizes social welfare. Unfortunately, it may result in poor revenue. To get revenue guarantees, knowledge of statistical information (e.g., in the form of the probability distributions of the agent valuations) can be very useful to the auction designer. Then, a usual trick is to design auctions that use reserve prices, which can secure a high payment by the highest bidder (or the auction winner, in general), even when the competition is low.

In this paper, we assume that no statistical information for the potential buyers is available. Instead, we assume that their valuations can have any value in a known interval $[1,H]$. An ambitious goal would be to design truthful auctions that always extract a revenue that is close to (i.e., a multiplicative approximation of) the highest valuation among the agents. For the second-price auction, this is an unrealistic goal; just imagine two bidders with valuations $1$ and $H$. Even worse, this is unrealistic for any deterministic auction, as the next argument shows. To get non-zero revenue from two bidders, both having a valuation of $1$, one should get the item at a price of at most $1$. Then, the monotonicity of truthful allocation rules implies that the same agent should get the item at the same price if her valuation was $H$ instead (and the other bidder still had a valuation of $1$).

To solve this issue keeping the ambitious goal of extracting a revenue that approximates the highest valuation, we resort to randomization. Randomized auctions can use allocation rules, which, for each bid profile, define the probability that each agent will get the item. They are still bound to Myerson's monotonicity, which is not restrictive anymore but rather reveals a much richer design space. For example, randomized auctions can exploit reserve prices, which would yield no revenue if used by deterministic auctions.

Furthermore, following the practice in the literature of learning-augmented algorithms and the model of~\cite{XL22} in particular, we assume that a {\em prediction} $\widehat{u}$ of the highest valuation is available. This means that both the range $[1,H]$ of valuations and the value $\widehat{u}$ can be hardwired in the definition of the auction, which should now satisfy two different revenue guarantees, known as {\em consistency} and {\em robustness} in the literature. First, for valuation profiles for which the prediction is correct (and the highest valuation is indeed equal to $\widehat{u}$), the auction should extract a $\gamma$-approximation of the highest valuation as revenue. For the remaining inputs, in which the prediction is incorrect, a typically worse $\rho$-approximation is sought. Such an auction will be said to have consistency and robustness of $\gamma$ and $\rho$, respectively.

Alternatively, we can allow the revenue guarantee to depend on the prediction error $\eta$, denoting how far (multiplicatively) the highest valuation in a given valuation profile is from the prediction $\widehat{u}$. Then, robustness is expressed as a function $\rho$, with $\rho(\eta)$ indicating the required lower bound on the revenue-over-highest-valuation ratio extracted by the auction in all valuation profiles with prediction error $\eta$. Hence, in the terminology of the setting of the previous paragraph, the value $\rho(1)$ is essentially a consistency guarantee.

\subsection{Our contribution} 
We study a class of randomized auctions with very appealing characteristics. We refer to them as {\em intuitive auctions}. An intuitive auction is anonymous (i.e., the outcome does not depend on the identifiers of the agents in any way) and, furthermore, has the property that only the highest bidder(s) get the item with positive probability. These properties make intuitive auctions particularly simple.

For the first setting, we give an optimal trade-off between the consistency and robustness of truthful (more precisely, of dominant-strategy incentive-compatible) intuitive auctions. As a corollary, we present auctions that have constant consistency and robustness $\Omega(\ln^{-1}{H})$. For the second setting, we present a sufficient and necessary condition for intuitive auctions with a given function $\rho$ of prediction error as robustness guarantee. As corollaries, we obtain robustness guarantees so that $1/\rho(\eta)$ is polylogarithmic in $\eta$ (even for unbounded valuations) or (sub)logarithmic, with a small dependency on $H$ as well. In our proofs, we use extensively a convenient new variation of Myerson's Lemma.

\subsection{Related work} Learning-augmented algorithms have emerged as a very hot topic nowadays, with numerous related contributions in recent years. Many classical problems have been reconsidered, and new algorithms, enhanced with (possibly erroneous) machine-learned predictions about their input, are designed and analyzed with respect to the consistency and robustness guarantees they can achieve. Representative problem domains include data structures, online and approximation algorithms for combinatorial optimization, streaming and sublinear algorithms, and many more. See the early survey by~\cite{MV20} as well as the online repository \url{algorithms-with-predictions.github.io}.

In algorithmic game theory and computational social choice, the concept of prediction has been considered for problems related to mechanism design~\citep{ABGOT22,BPS23,BGT23,BGTZ23,IB22,XL22}, price of anarchy of cost sharing~\citep{GKST22}, and distortion of voting~\citep{BFGT23}. The paper by~\cite{XL22} is closest to ours. Among other problems, the authors study single-item auctions with the revenue-over-highest valuation objective and present consistency and robustness bounds for deterministic auctions. In particular, they present a truthful auction which is proved to have consistency $\gamma$ and per-instance robustness that is a function of $\gamma$, the prediction error, and the upper bound on the valuation of agents. We remark that the use of randomization in the current paper allows us to obtain considerably better consistency vs.~robustness tradeoffs compared to the results of \cite{XL22}.

Our assumptions of agents with worst-case valuations are similar in spirit to early work on competitive auctions, initiated with the work of~\cite{GHKSW06}. Even though revenue has been the main concern in that line of research, weaker benchmarks than the highest valuation among all agents have been mainly considered. See the related discussion in the very recent paper by~\cite{LWZ23}, who study competitive auctions with predictions. The more distant but also extremely important field of Bayesian mechanism design, which uses extensively statistical information about the agent valuations, is surveyed by~\cite{H13}.

\section{Preliminaries}\label{sec:prelim}
We begin with some background on auctions and mechanism design. The interested reader may find more information in standard textbooks, e.g., see \citet[Chapter 3]{R16}.
In single-item auctions, a set of $n$ {\em agents} (or {\em bidders}) compete for an item. Each agent has a private {\em valuation} $v_i$ for the item. All valuations belong to the interval $[1,H]$ for $H>1$. An {\em auction mechanism} (or, simply, auction) receives {\em bids} from the agents as input (to be thought of as reportings of their valuations) and decides the agent who will get the item (or that no agent should get the item) and the payment that will be received from each agent. Auctions are, in general, randomized. Formally, the auction consists of an {\em allocation rule} $\x=(x_1, x_2, ..., x_n)$ and a {\em payment rule} $\p=(p_1, p_2, ..., p_n)$. Both receive the {\em bid vector} $\bb=(b_1, b_2, ..., b_n)$ as input, consisting of one bid per agent. The allocation function $x_i(\bb)$ denotes the probability that agent $i$ gets the item while the payment $p_i(\bb)$ denotes the payment by agent $i$. Thus, $x_i(\bb)\in [0,1]$ and $p_i(\bb)\geq 0$. An allocation rule $\x$ is {\em feasible} if $\sum_{i\in [n]}{x_i(\bb)}\leq 1$ for every bid vector $\bb$.

Agents are (expected) utility maximizers. Agent $i$'s utility, from the outcome of an auction that uses the allocation rule $\x$ and the payment rule $\p$ when the agents submit the bid vector $\bb$, is defined as $u_i(\bb)=v_i\cdot x_i(\bb)-p_i(\bb)$. Ideally, we would like to use auctions that motivate agents to report their private valuations {\em truthfully} as bids. Truthfulness requires that each agent $i$ maximizes her utility by reporting her valuation as bid. Formally, this requirement can be written as $u_i(v_i,\bb_{-i})\geq u_i(z,\bb_{-i})$, for every bid vector $\bb_{-i}$ submitted by the agents different than $i$ and for every possible bid $z\in [1,H]$ agent $i$ may consider to submit. In addition, the property of {\em individual rationality} requires that every truthful agent $i$ has non-negative utility $u_i(v_i,\bb_{-i})$, for any possible bids $\bb_{-i}$ by other agents. An auction that is both truthful and individually rational is called {\em dominant-strategy incentive-compatible} (DSIC).

One of the most fundamental results in auction theory is the characterization of DSIC auctions by~\citet{M81}, which connects the property of truthfulness to the monotonicity of the allocation rule. An allocation rule $\x$ is {\em monotone} if, for every agent $i$ and any possible bids $\bb_{-i}$ by the other agents, the function $x_i(z,\bb_{-i})$ is non-decreasing in terms of $z$. We say that an allocation rule $\x$ is {\em implementable} if there exist a payment rule $\p$ so that the auction consisting of these two rules is DSIC.

\begin{lemma}[Myerson's Lemma]\label{lem:myerson}
The allocation rule $\x$ is implementable if and only if it is monotone. If $\x$ is monotone, it is implementable through the payment rule $\p$, which, for every agent $i$ and any bids $\bb_{-i}$ by the other agents, it holds $x_i(1,\bb_{-i})\geq p_i(1,\bb_{-i})$ and
\begin{align*}
    p_i(t,b_{-i}) - p_i(s,b_{-i})
    &= t\cdot x_i(t,b_{-i})-s\cdot x_i(s,b_{-i})-\int_s^t{x_i(z,b_{-i})\ud{z}}  
\end{align*}
for every $s,t\in [1,H]$ with $s\leq t$.
\end{lemma}
We now give an alternative to Myerson's Lemma, which will be more convenient for our proofs.

\begin{lemma}\label{lem:myerson-alt}
A monotone allocation rule $\x$ is implementable through a payment rule $\p$ if and only if, for every agent $i$ and any bids $\bb_{-i}$ by the other agents, it holds $x_i(1,\bb_{-i})\geq p_i(1,\bb_{-i})$ and
\begin{align*}
    x_i(t,b_{-i}) &= \frac{p_i(t,b_{-i})}{t}+\int_s^t{\frac{p_i(z,b_{-i})}{z^2}\ud{z}}+x_i(s,b_{-i})-\frac{p_i(s,b_{-i})}{s}
\end{align*}    
for every $s,t\in [1,H]$ with $s\leq t$.
\end{lemma}

Before proving Lemma~\ref{lem:myerson-alt}, we remark that the formulas connecting the allocation and payment functions in Lemmas~\ref{lem:myerson} and~\ref{lem:myerson-alt} characterize truthfulness while the inequality $x_i(1,\bb_{-i})\geq p_i(1,\bb_{-i})$ yields individual rationality. Notice that neither the particular version of Myerson's Lemma that we use here nor our Lemma~\ref{lem:myerson-alt} make any differentiability assumptions for the allocation or payment functions.

\begin{proof}
It suffices to show that the two different formulas in the statements of Lemmas~\ref{lem:myerson} and~\ref{lem:myerson-alt} are equivalent.

For valuations $s,t\in [1,H]$ with $s\leq t$, let \begin{align*}
\sigma_s(t)&=\frac{\int_s^t{x_i(z,b_{-i})\ud{z}}}{t}.
\end{align*}
By applying Myerson's Lemma, we get
\begin{align*}
    p_i(t,b_{-i})+s\cdot x_i(s,b_{-i})-p(s,b_{-i}) &= t\cdot x_i(t,b_{-i})-\int_s^t{x_i(z,b_{-i})\ud{z}}=t^2\cdot \sigma'_s(t)
\end{align*}
and, equivalently, 
\begin{align*}
    \sigma'_s(t) &= \frac{p_i(t,b_{-i})}{t^2}+\frac{s\cdot x_i(s,b_{-i})-p_i(s,b_{-i})}{t^2}.
\end{align*}
Thus,
\begin{align*}
    \sigma_s(t) &=\sigma_s(s)+\int_s^t{\frac{p_i(z,b_{-i})}{z^2}\ud{z}}+\left(\frac{1}{s}-\frac{1}{t}\right)\cdot(s\cdot x_i(s,b_{-i})-p_i(s,b_{-i})),
\end{align*}
implying, by the definition of $\sigma_s(t)$, that
\begin{align*}
\int_s^t{x_i(z,b_{-i})\ud{z}} &= t\int_s^t{\frac{p_i(z,b_{-i})}{z^2}\ud{z}}+\left(\frac{t}{s}-1\right)\cdot(s\cdot x_i(s,b_{-i})-p_i(s,b_{-i})).
\end{align*}
The lemma now follows after differentiating both sides with respect to $t$.
\end{proof}

We consider {\em anonymous} auctions, in which the allocation function $x_i(z,\bb_{-i})$ and the payment function $p_i(z,\bb_{-i})$ do not depend on $i$. Furthermore, we consider auctions which give the item with positive probability to the highest bidder(s) only, and this probability depends on the highest and second highest bid. Due to anonymity, ties are resolved uniformly at random among the tied agents. We use the term {\em intuitive} to refer to such auctions. We also simplify notation for bid vectors, and allocation and payment functions as follows. We describe a bid vector $\bb$ as the triplet $(t,b,\nu)$, where $t$ is the bid of a particular agent $i$, $b$ denotes the highest bid among the other agents, and $\nu$ denotes the number of agents different than $i$ with a bid of $b$. Thus, anonymity allows us to use $x(t,b,\nu)$ and $p(t,b,\nu)$ to refer to allocations and payments in general. Of course, for auctions with a single agent, which will be important in our study, the bid of the agent is the only parameter required to describe the bid vector; we use $x(t)$ and $p(t)$ to refer to allocation and payments in this case.

We use the term {\em revenue} to refer to the total payment received by all agents and aim to design auctions that extract a high revenue. Ideally, an auction applied on agents with a highest valuation of $t$ should extract a revenue of as close to $t$ as possible. We consider the {\em revenue-over-highest-valuation ratio} as our main objective. We assume that in addition to parameter $H$ denoting the upper bound on valuations, we are given a {\em prediction} $\widehat{u}$ for the highest valuation. Then, for parameters $\gamma$ and $\rho$, we aim to design auctions with revenue-over-highest valuation guarantees of $\gamma$ and $\rho$ when the prediction is correct and incorrect, respectively. Then, such auctions will be said to have {\em consistency} of $\gamma$ and {\em robustness} of $\rho$. Alternatively, we can allow the robustness requirement to be a function of the {\em prediction error} $\max\{t/\widehat{u},\widehat{u}/t\}$ indicating how far the highest valuation $t$ is from the prediction $\widehat{u}$. The robustness requirement can then be described by a non-increasing function $\rho:[1,H]\rightarrow [0,1]$, asking for a revenue that is at least $\rho\left(\max\{t/\widehat{u},\widehat{u}/t\}\right)\cdot t$.

\section{A consistency vs.~robustness trade-off}
We devote this section to proving the following statement, which gives sufficient and necessary conditions for the existence of consistent and robust auctions.

\begin{theorem}\label{thm:single-bidder}
    Let $0\leq \rho\leq \gamma\leq 1$. There exists a $\gamma$-consistent and $\rho$-robust DSIC intuitive auction for bidders with valuations from $[1,H]$ and a prediction for the highest bid $\widehat{u}\in [1,H]$ if and only if $\gamma+\rho\cdot \ln{\max\left\{\widehat{u},\frac{H\cdot \rho}{
    \gamma}\right\}}\leq 1$.
\end{theorem}
The proof of Theorem~\ref{thm:single-bidder} is given in Sections~\ref{subsec:proof-if} and~\ref{subsec:proof-only-if} below. As a corollary of Theorem~\ref{thm:single-bidder}, we obtain that auctions with constant consistency and robustness $\Omega(\ln^{-1}{H})$ do exist. E.g., notice that the parameters $\gamma=1/2$ and $\rho=\frac{1}{2(1+\ln{H})}$ satisfy the condition of Theorem~\ref{thm:single-bidder}. Even though Theorem~\ref{thm:single-bidder} just claims the existence of the desired auction, its proof is constructive. Once we have parameters $\gamma$ and $\rho$ satisfying the condition in the theorem, the definition of a corresponding auction is given explicitly in Section~\ref{subsec:proof-if} below.

\subsection{Proving the ``if'' part of Theorem~\ref{thm:single-bidder}}\label{subsec:proof-if}

We prove the ``if'' part of Theorem~\ref{thm:single-bidder} by constructing a DSIC intuitive auction which uses the allocation function $\bar{x}$ and the payment function $\bar{p}$ defined below. Consider an agent $i$ with valuation $t$. Let $b$ be the highest bid among the remaining agents, and let $\nu$ denote the number of bidders different than $i$ with valuation $b$.

The allocation fraction $\bar{x}(t,b,\nu)$ is defined as
\begin{align*}
    \bar{x}(t,b,\nu) &= \begin{cases} 0, &t\in [1,b)\\
    \frac{\rho}{\nu+1}, &t=b\\
    \rho+\rho\cdot \ln\frac{t}{b}, &t\in (b,H]
    \end{cases}
\end{align*}
if $1\leq \widehat{u}<b\leq H$, as
\begin{align*}
\bar{x}(t,b,\nu) &= \begin{cases} 0, &t\in [1,\widehat{u})\\
\frac{\gamma}{\nu+1}, &t=\widehat{u}\\
\gamma, &t\in \left(\widehat{u},\min\left\{\frac{\gamma\cdot \widehat{u}}{\rho}, H\right\}\right]\\
\gamma+\rho\cdot \ln{\frac{t\cdot \rho}{b\cdot \gamma}}, &t\in \left(\min\left\{\frac{\gamma\cdot \widehat{u}}{\rho}, H\right\},H\right]
\end{cases}
\end{align*}
if $1\leq \widehat{u}=b\leq H$, and as 
\begin{align*}
\bar{x}(t,b,\nu) &= \begin{cases} 0, &t\in [1,b)\\
\frac{\rho}{\nu+1}, &t=b\\
\rho+\rho\cdot \ln\frac{t}{b}, &t\in (b,\widehat{u})\\
\gamma+\rho\cdot \ln{\frac{\widehat{u}}{b}}, &t\in \left[\widehat{u},\min\left\{\frac{\gamma\cdot \widehat{u}}{\rho}, H\right\}\right]\\
\gamma+\rho\cdot \ln{\frac{t\cdot \rho}{b\cdot \gamma}}, &t\in \left(\min\left\{\frac{\gamma\cdot \widehat{u}}{\rho}, H\right\},H\right]
\end{cases}
\end{align*}
if $1\leq b<\widehat{u}\leq H$. See Figure~\ref{fig:allocation}.

\begin{figure}[h]
    \centering
    \includegraphics[scale=0.6]{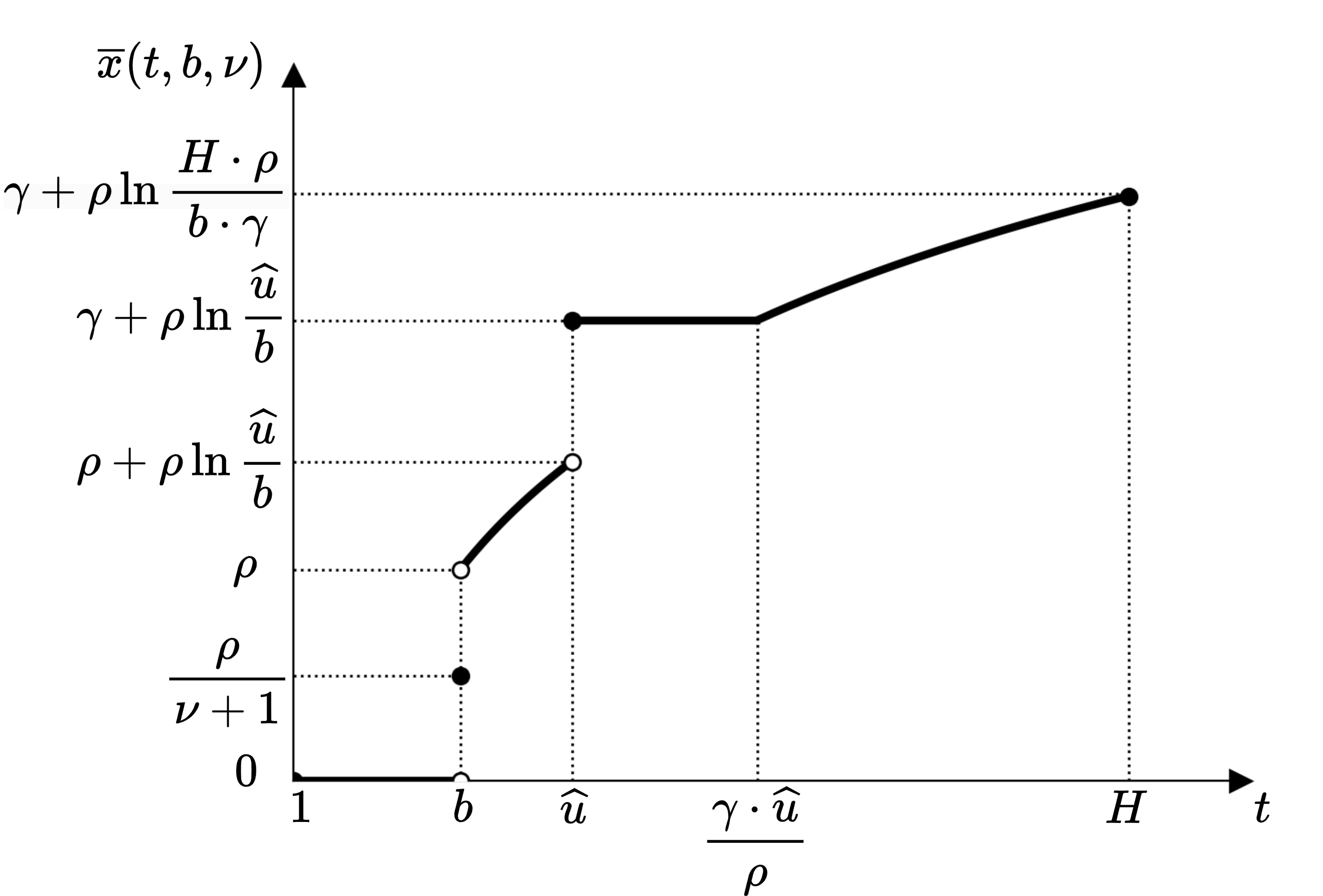}
    \caption{The most general form of the allocation function $\bar{x}$ for an agent in terms of her valuation $t$, assuming $1<b<\widehat{u}<\frac{\gamma\cdot \widehat{u}}{\rho}<H$ and $\nu=1$. Notice that $\bar{x}(t,b,\nu)$ consists of five parts: the leftmost part in which the item is not allocated to the agent, the point corresponding to a tie for the highest bid, and three more parts in which the allocation function has logarithmic, constant, and again logarithmic form. The remaining cases for the  relative values of $b$, $\widehat{u}$, $\frac{\gamma\cdot \widehat{u}}{\rho}$ and $H$ do not include some of the three rightmost parts. The black and white dots are used at points in which the allocation function ``jumps''; the black dot represents the allocation value at these points.}
    \label{fig:allocation}
\end{figure}

Notice that, in all cases, $\bar{x}$ is non-negative and non-decreasing in $t$. Whenever agent $i$ is tied as highest bidder with a value of $t=b$  (together with $\nu$ more agents), the total fraction allocated is $\gamma$ when $b=\widehat{u}$ and $\rho$ otherwise, i.e., at most $1$. When agent $i$ is the unique highest bidder, she is the only agent who gets a positive fraction of the item, which is maximized to either $\rho+\rho\cdot \ln\frac{H}{b}$, or $\gamma+\rho\cdot \ln\frac{\widehat{u}}{b}$, or $\gamma+\rho\cdot \ln\frac{H\cdot \rho}{b\cdot \gamma}$. We have
\begin{align*}
    \rho+\rho\cdot \ln\frac{H}{b} &\leq \rho+\rho\cdot \ln{H} \leq \rho+\rho\cdot \ln{H}+\rho\cdot \left(\frac{\gamma}{\rho}-1-\ln\frac{\gamma}{\rho}\right)\leq \gamma+\rho\cdot \ln{\max\left\{\widehat{u},\frac{H\cdot \rho}{
    \gamma}\right\}} \leq 1.
\end{align*}
The first two inequalities follow since $b\geq 1$ and using the inequality $z-1-\ln{z}\geq 0$ for $z>0$. Furthermore, 
\begin{align*}
\max\left\{\gamma+\rho\cdot \ln\frac{\widehat{u}}{b},\gamma+\rho\cdot \ln\frac{H\cdot \rho}{b\cdot \gamma}\right\}&\leq \gamma+\rho\cdot \ln{\max\left\{\widehat{u},\frac{H\cdot \rho}{
    \gamma}\right\}}\leq 1,
\end{align*}
by the assumption of Theorem~\ref{thm:single-bidder}. Thus, the allocation function $\bar{x}$ is feasible and monotone. 

The corresponding payment function $\bar{p}$ is defined as
\begin{align*}
    \bar{p}(t,b,\nu) &= \begin{cases} 0, &t\in [1,b)\\
    \frac{\rho\cdot b}{\nu+1}, &t=b\\
    \rho\cdot t, &t\in (b,H]
    \end{cases}
\end{align*}
if $1\leq \widehat{u}<b\leq H$, as
\begin{align*}
\bar{p}(t,b,\nu) &= \begin{cases} 0, &t\in [1,\widehat{u})\\
\frac{\gamma\cdot \widehat{u}}{\nu+1}, &t=\widehat{u}\\
\gamma \cdot \widehat{u}, &t\in \left(\widehat{u},\min\left\{\frac{\gamma\cdot \widehat{u}}{\rho}, H\right\}\right]\\
\rho\cdot t, &t\in \left(\min\left\{\frac{\gamma\cdot \widehat{u}}{\rho}, H\right\},H\right]
\end{cases}
\end{align*}
if $1\leq \widehat{u}=b\leq H$, and as 
\begin{align*}
\bar{p}(t,b,\nu) &= \begin{cases} 0, &t\in [1,b)\\
\frac{\rho\cdot b}{\nu+1}, &t=b\\
\rho\cdot b, &t\in (b,\widehat{u})\\
\gamma\cdot \widehat{u}, &t\in \left[\widehat{u},\min\left\{\frac{\gamma\cdot \widehat{u}}{\rho}, H\right\}\right]\\
\rho\cdot t, &t\in \left(\min\left\{\frac{\gamma\cdot \widehat{u}}{\rho}, H\right\},H\right]
\end{cases}
\end{align*}
if $1\leq b<\widehat{u}\leq H$.

It is straightforward to verify that, indeed, the payment function $\bar{p}$ implements the allocation function $\bar{x}$, i.e., that $\bar{x}$ and $\bar{p}$ are consistent with Lemma~\ref{lem:myerson-alt}. Furthermore, we can verify that the auction's revenue yields the required consistency and robustness. When the prediction is correct, and the highest bid is equal to $\widehat{u}$, the revenue is exactly $\gamma\cdot \widehat{u}$, which implies a consistency of $\gamma$. When the highest bid $t$ is different than $\widehat{u}$, the revenue is either equal to $\rho\cdot t$ or equal to $\gamma\cdot \widehat{u}$ for $t\in \left(\widehat{u},\min\left\{\frac{\gamma\cdot \widehat{u}}{\rho}, H\right\}\right]$, i.e., for $t\leq \frac{\gamma\cdot \widehat{u}}{\rho}$. Notice that $\gamma\cdot \widehat{u}\geq \rho\cdot t$ in this case, implying the bound of $\rho$ for robustness. The completes the proof of the ``if'' part of Theorem~\ref{thm:single-bidder}.

\subsection{Proving the ``only if'' part of Theorem~\ref{thm:single-bidder}}\label{subsec:proof-only-if}

We will now prove the ``only if'' part of Theorem~\ref{thm:single-bidder} by considering the case of a single bidder of valuation $t$. We will use the simplest univariate form of allocation and payment functions. For the sake of contradiction, let us assume that 
\begin{align}\label{eq:ass1}
\gamma+\rho\cdot \ln{\max\left\{\widehat{u},\frac{H\cdot \rho}{\gamma}\right\}}&>1
\end{align}
and, furthermore, that there is a feasible allocation rule $x$ which is implementable through a payment function $p$ that satisfies $p(t)\geq \rho\cdot t$ for $t\in [1,H]$ and $p(\widehat{u})\geq \gamma\cdot \widehat{u}$. I.e., the auction defined by the pair $(x,p)$ is $\gamma$-consistent and $\rho$-robust.

By the robustness requirement, we have $p(t)\geq \rho\cdot t$ for $t\in [1,\widehat{u}]$ and, hence,
\begin{align*}
    \int_1^{\widehat{u}}{\frac{p(z)}{z^2}\ud{z}}\geq \rho \int_1^{\widehat{u}}{\frac{\ud{z}}{z}}=\rho\cdot \ln{\widehat{u}}.
\end{align*}
By applying Lemma~\ref{lem:myerson-alt}, we have the individual rationality condition $x(1)-p(1)\geq 0$ and, furthermore, for $s=1$ and $t=\widehat{u}$,
\begin{align}\label{eq:x_u}
x(\widehat{u}) &= \frac{p(\widehat{u})}{\widehat{u}}+\int_1^{\widehat{u}}{\frac{p(z)}{z^2}\ud{z}}+x(1)-p(1) \geq \frac{p(\widehat{u})}{\widehat{u}}+\rho\cdot \ln{\widehat{u}}.
\end{align}

We now distinguish between two cases. First, if $\widehat{u}\geq \frac{H\cdot\rho}{\gamma}$, inequality (\ref{eq:x_u}) yields (using the consistency requirement $p(\widehat{u})\geq \gamma\cdot \widehat{u}$) 
\begin{align*}
    x(\widehat{u}) &\geq \gamma+\rho\cdot \ln{\widehat{u}} = \gamma+\rho\cdot \ln{\max\left\{\widehat{u},\frac{H\cdot \rho}{\gamma}\right\}}>1.
\end{align*}
The strict inequality follows by assumption (\ref{eq:ass1}) and contradicts the feasibility of the allocation function $x$.

In the following, we consider the second case, where $\widehat{u}< \frac{H\cdot \rho}{\gamma}$. Define $v=\frac{\gamma\cdot \widehat{u}}{\rho}$ and observe that $v\in [\widehat{u},H)$. 

By the consistency requirement and payment monotonicity, we have $p(t)\geq \gamma\cdot \widehat{u}$ for $t\in [\widehat{u},v)$ and, hence,
\begin{align*}
\int_{\widehat{u}}^v{\frac{p(z)}{z^2}\ud{z}} &\geq \gamma\cdot \widehat{u}\int_{\widehat{u}}^v{\frac{\ud{z}}{z^2}}=\gamma-\frac{\gamma\cdot \widehat{u}}{v}=\gamma-\rho.
\end{align*}
Thus, by applying Lemma~\ref{lem:myerson-alt} for $s=\widehat{u}$ and $t=v$, we get
\begin{align}\label{eq:x_v}
x(v) &= \frac{p(v)}{v}+\int_{\widehat{u}}^{v}{\frac{p(z)}{z^2}\ud{z}}+x(\widehat{u})-\frac{p(\widehat{u})}{\widehat{u}}\geq \frac{p(v)}{v}+\gamma-\rho+x(\widehat{u})-\frac{p(\widehat{u})}{\widehat{u}}.
\end{align}

Also, by the robustness requirement, we have $p(t)\geq \rho\cdot t$ for $t\in [v,H]$ and, hence (using the definition of $v$),
\begin{align*}
    \int_v^H{\frac{p(z)}{z^2}\ud{z}} &\geq \rho\int_v^H{\frac{\ud{z}}{z}}=\rho\cdot \ln{\frac{H}{v}}=\rho\cdot \ln{\frac{H\cdot \rho}{\widehat{u}\cdot \gamma}}.
\end{align*}
Thus, by applying Lemma~\ref{lem:myerson-alt} for $s=v$ and $t=H$ and using the robustness requirement $p(H)\geq \rho\cdot H$, we get
\begin{align}\label{eq:x_H}
x(H) &= \frac{p(H)}{H}+\int_v^H{\frac{p(z)}{z^2}\ud{z}}+x(v)-\frac{p(v)}{v}
\geq \rho+\rho\cdot \ln{\frac{H\cdot \rho}{\widehat{u}\cdot \gamma}}+x(v)-\frac{p(v)}{v}.
\end{align}

Now, by summing equations (\ref{eq:x_u}), (\ref{eq:x_v}), and (\ref{eq:x_H}), and using the assumption $\widehat{u}<\frac{H\cdot \rho}{\gamma}$ in the second case and assumption (\ref{eq:ass1}), we obtain
\begin{align*}
    x(H) &\geq \gamma+\ln{\frac{H\cdot \rho}{\gamma}}=\gamma+\ln{\max\left\{\widehat{u},\frac{H\cdot \rho}{\gamma}\right\}}>1,
\end{align*}
again contradicting the feasibility of the allocation function. This completes the proof of the ``only if'' part of Theorem~\ref{thm:single-bidder}.

\section{Robustness as a function of prediction error}\label{sec:pred}
In this section, we give the formal statement and proof of our second result. 
\begin{theorem}\label{thm:pred-error}
    Let $\rho:[1,H]\rightarrow [0,1]$ be a differentiable function so that $\rho(\eta)$ is non-increasing and $\eta\cdot \rho(\eta)$ is non-decreasing in $\eta$. There exists a $\rho$-robust DSIC intuitive auction for bidders with valuations from $[1,H]$ and a prediction for the highest bid $\widehat{u}\in [1,H]$ if and only if 
    $$\rho(H/\widehat{u})+\int_1^{\widehat{u}}{\frac{\rho(z)}{z} \ud{z}}+\int_1^{H/\widehat{u}}{\frac{\rho(z)}{z} \ud{z}}\leq 1.$$
\end{theorem}
The proof follows in Sections~\ref{subsec:proof-pred-if} and~\ref{subsec:proof-pred-only-if}. Then, we present specific robustness guarantees (i.e., functions of prediction error) as applications in Section~\ref{subsec:pred-appl}.

\subsection{Proving the ``if'' part of Theorem~\ref{thm:pred-error}}\label{subsec:proof-pred-if}
We prove the ``if'' part of Theorem~\ref{thm:pred-error} by constructing a DSIC intuitive auction, which uses the allocation function $\widetilde{x}$ and the payment function $\widetilde{p}$ defined below. Of course, the definition of both $\widetilde{x}$ and $\widetilde{p}$ uses the function $\rho$. 

Again, consider an agent $i$ with valuation $t$. Let $b$ be the highest bid among the remaining agents, and let $\nu$ denote the number of bidders different than $i$ with valuation $b$. The allocation fraction $\widetilde{x}(t,b,\nu)$ is defined as
\begin{align*}
    \widetilde{x}(t,b,\nu) &= \begin{cases} 0, &t\in [1,b)\\
    \frac{\rho(b/\widehat{u})}{\nu+1}, &t=b\\ \rho(t/\widehat{u})+\int_{b/\widehat{u}}^{t/\widehat{u}}{\frac{\rho(z)}{z}\ud{z}}, &t\in (b,H]
    \end{cases}
\end{align*}
if $1\leq \widehat{u}\leq b\leq H$, and as
\begin{align*}
\widetilde{x}(t,b,\nu) &= \begin{cases} 0, &t\in [1,b)\\
\frac{\rho(\widehat{u}/b)}{\nu+1}, &t=b\\
\rho(\widehat{u}/t)+\int_{\widehat{u}/t}^{\widehat{u}/b}{\frac{\rho(z)}{z}\ud{z}}, &t\in (b,\widehat{u})\\
\rho(t/\widehat{u})+\int_{1}^{\widehat{u}/b}{\frac{\rho(z)}{z}\ud{z}}\\
\quad +\int_{1}^{t/\widehat{u}}{\frac{\rho(z)}{z}\ud{z}}, &t\in [\widehat{u},H]\end{cases}
\end{align*}
if $1\leq b<\widehat{u}\leq H$.

The quantity $\widetilde{x}(t,b,\nu)$ is clearly non-negative. We will show that it is also non-decreasing. Let $\sigma_1(\eta)=\eta\cdot \rho(\eta)$ and $\sigma_2(\eta)=\rho(\eta)/\eta$. By the definition of the allocation function for $t$ such that $1\leq \widehat{u}\leq b<t\leq H$ and $1\leq b<\widehat{u}\leq t\leq H$, its derivative with respect to $t$ is
\begin{align}\nonumber
\widetilde{x}'(t,b,\nu) &= \frac{1}{t}\cdot\left(\frac{t}{\widehat{u}}\cdot \rho'(t/\widehat{u})+\rho(t/\widehat{u})\right)\\\label{eq:sigma_1}
&= \frac{1}{t}\cdot \sigma'_1(t/\widehat{u}).
\end{align}
Also, for $t$ such that $1\leq b<t<\widehat{u}\leq H$, the derivative of the allocation function with respect to $t$ is
\begin{align}\nonumber
\widetilde{x}'(t,b,\nu) &= -\frac{1}{t}\cdot\left(\frac{\widehat{u}}{t}\cdot \rho'(\widehat{u}/t)-\rho(\widehat{u}/t)\right)\\\label{eq:sigma_2}
&= -\frac{u^2}{t^3}\cdot \sigma'_2(\widehat{u}/t).
\end{align}
Furthermore, notice that, in any case, function $\widetilde{x}(t,b,\nu)$ is continuous at $t=\widehat{u}$. Thus, equations (\ref{eq:sigma_1}) and (\ref{eq:sigma_2}) imply that $\widetilde{x}(t,b,\nu)$ is indeed non-decreasing (recall that $\sigma_1(\eta)$ is non-decreasing and $\sigma_2(\eta)$ is non-increasing in $\eta$ and their derivatives are non-negative and non-positive, respectively). To show feasibility, it suffices to show that the highest fraction of the item that is allocated does not exceed $1$. Indeed, if $1\leq \widehat{u}\leq b\leq H$, by the definition of the allocation function and the condition of Theorem~\ref{thm:pred-error}, we have 
\begin{align*}
\widetilde{x}(H,b,\nu) &= \rho(H/\widehat{u})+\int_{b/\widehat{u}}^{H/\widehat{u}}{\frac{\rho(z)}{z}\ud{z}}\leq 
\rho(H/\widehat{u})+\int_{1}^{H/\widehat{u}}{\frac{\rho(z)}{z}\ud{z}} \leq 1,
\end{align*}    
while, if $1\leq b<\widehat{u}\leq H$, we have 
\begin{align*}
\widetilde{x}(H,b,\nu) &= \rho(H/\widehat{u})+\int_{1}^{\widehat{u}/b}{\frac{\rho(z)}{z}\ud{z}}+\int_{1}^{H/\widehat{u}}{\frac{\rho(z)}{z}\ud{z}}\leq 
\rho(H/\widehat{u})+\int_{1}^{\widehat{u}}{\frac{\rho(z)}{z}\ud{z}} \int_{1}^{H/\widehat{u}}{\frac{\rho(z)}{z}\ud{z}} \leq 1,
\end{align*}    
as desired. Thus, the allocation function $\widetilde{x}$ is feasible and monotone.

The corresponding payment function $\widetilde{p}$ is defined as
\begin{align*}
    \widetilde{p}(t,b,\nu) &= \begin{cases} 0, &t\in [1,b)\\
    \frac{\rho(b/\widehat{u})\cdot b}{\nu+1}, &t=b\\
    \rho(t/\widehat{u})\cdot t, &t\in (b,H]
    \end{cases}
\end{align*}
if $1\leq \widehat{u}\leq b\leq H$, and as
\begin{align*}
\widetilde{p}(t,b,\nu) &= \begin{cases} 0, &t\in [1,b)\\
\frac{\rho(\widehat{u}/b)\cdot b}{\nu+1}, &t=b\\
\rho(\widehat{u}/t)\cdot t, &t\in (b,\widehat{u})\\
\rho(t/\widehat{u})\cdot t, &t\in [\widehat{u},H]
\end{cases}
\end{align*}
if $1\leq b<\widehat{u}\leq H$. 

We will show that $\widetilde{p}$ indeed implements the allocation rule $\widetilde{x}$. i.e., $\widetilde{x}$ and $\widetilde{p}$ are consistent with Lemma~\ref{lem:myerson-alt}. For $t$ such that $1\leq \widehat{u}\leq b\leq t\leq H$, the RHS of the expression in Lemma~\ref{lem:myerson-alt} for $t$ and $s=b$ becomes
\begin{align*}
\frac{p(t)}{t}+\int_b^t{\frac{p(z)}{z^2}\ud{z}}+x(b)-\frac{p(b)}{b} &=\rho(t/\widehat{u})+\int_b^t{\frac{\rho(z/\widehat{u})}{z}\ud{z}}=\rho(t/\widehat{u})+\int_{t/\widehat{u}}^{b/\widehat{u}}{\frac{\rho(y)}{y}\ud{y}}=x(t),
\end{align*}
as desired. The second equality follows using the substitution $y=z/\widehat{u}$. For $t$ such that $1\leq b< t< \widehat{u}\leq H$, the RHS of the expression in Lemma~\ref{lem:myerson-alt} for $t$ and $s=b$ becomes
\begin{align*}
\rho(\widehat{u}/t)+\int_b^t{\frac{\rho(\widehat{u}/z)}{z}\ud{z}}=\rho(\widehat{u}/t)+\int_{\widehat{u}/t}^{\widehat{u}/b}{\frac{\rho(y)}{y}\ud{y}}=x(t),
\end{align*}
as desired. The first equality follows using the substitution $y=\widehat{u}/z$. Finally, for $t$ such that $1\leq b< \widehat{u}\leq t\leq H$, the RHS of the expression in Lemma~\ref{lem:myerson-alt} for $t$ and $s=\widehat{u}$ becomes
\begin{align*}
\frac{p(t)}{t}+\int_{\widehat{u}}^t{\frac{p(z)}{z^2}\ud{z}}+x(\widehat{u})-\frac{p(\widehat{u})}{\widehat{u}}
&=\rho(t/\widehat{u})+\int_{\widehat{u}}^t{\frac{\rho(z/\widehat{u})}{z}\ud{z}}+\int_1^{\widehat{u}/b}{\frac{\rho(z)}{z}\ud{z}}\\
&=\rho(t/\widehat{u})+\int_1^{t/\widehat{u}}{\frac{\rho(y)}{y}}+\int_1^{\widehat{u}/b}{\frac{\rho(z)}{z}\ud{z}}=x(t),
\end{align*}
as desired, again. The second equality follows using the substitution $y=z/\widehat{u}$.

Notice that, clearly, the auction's revenue yields the required robustness. In particular, assuming, without loss of generality, that agent $i$ has the highest valuation (i.e., $t\geq b$), the revenue is exactly $\rho(\widetilde{u}/t)\cdot t$ for $t\in [b,\widehat{u})$ and exactly $\rho(t/\widetilde{u})\cdot t$ for $t\in [\widehat{u},H]$. The proof of the ``if'' part of Theorem~\ref{thm:pred-error} is now complete.

\subsection{Proving the ``only if'' part of Theorem~\ref{thm:pred-error}}\label{subsec:proof-pred-only-if}

Like in the corresponding proof for Theorem~\ref{thm:single-bidder}, we will prove the ``only if'' part of Theorem~\ref{thm:pred-error} by considering the case of a single bidder. For the sake of contradiction, let us assume that 
\begin{align}\label{eq:ass2}
\rho(H/\widehat{u})+\int_1^{\widehat{u}}{\frac{\rho(z)}{z} \ud{z}}+\int_1^{H/\widehat{u}}{\frac{\rho(z)}{z} \ud{z}}& > 1
\end{align}
and, furthermore, that there exists a feasible and monotone allocation function $x$ which is implementable through a payment function $p$ that satisfies $p(t)\geq \rho(\widehat{u}/t)\cdot t$ for $t\in [1,\widehat{u}]$ and $p(t)\geq \rho(t/\widehat{u})\cdot t$ for $t\in [\widehat{u},H]$.

By applying Lemma~\ref{lem:myerson-alt}, we have the individual rationality condition $x(1)-p(1)\geq 0$ and, for $s=1$ and $t=\widehat{u}$, 
\begin{align}\label{eq:x_u-pred-err}
    x(\widehat{u}) &= \frac{p(\widehat{u})}{\widehat{u}}+\int_1^{\widehat{u}}{\frac{p(z)}{z^2}\ud{z}}+x(1)-p(1) \geq \frac{p(\widehat{u})}{\widehat{u}}+\int_1^{\widehat{u}}{\frac{\rho(\widehat{u}/z)}{z}\ud{z}}.
\end{align}
By applying Lemma~\ref{lem:myerson-alt} for $s=\widehat{u}$ and $t=H$ and using our assumptions for the payment function $p$, 
we get
\begin{align}\label{eq:x_H-pred-err}
    x(H) &= \frac{p(H)}{H}+\int_{\widehat{u}}^H{\frac{p(z)}{z^2}\ud{z}}+x(\widehat{u})-\frac{p(\widehat{u})}{\widehat{u}}\geq \rho(H/\widehat{u})+\int_{\widehat{u}}^H{\frac{\rho(z/\widehat{u})}{z}\ud{z}}+x(\widehat{u})-\frac{p(\widehat{u})}{\widehat{u}}.
\end{align}
By inequalities (\ref{eq:x_u-pred-err}) and (\ref{eq:x_H-pred-err}), we have
\begin{align*}
    x(H) &\geq \rho(H/\widehat{u})+\int_1^{\widehat{u}}{\frac{\rho(\widehat{u}/z)}{z}\ud{z}}+\int_{\widehat{u}}^H{\frac{\rho(z/\widehat{u})}{z}\ud{z}}=\rho(H/\widehat{u})+\int_1^{\widehat{u}}{\frac{\rho(y)}{y}\ud{y}}+\int_1^{H/\widehat{u}}{\frac{\rho(y)}{y}\ud{y}}>1,
\end{align*}
contradicting the feasibility of the allocation function $x$. The equality follows by the substitution $y=\widehat{u}/z$ and $y=z/\widehat{u}$ in the two integrals and the second inequality uses our assumption (\ref{eq:ass2}). The proof of the ``only if'' part of Theorem~\ref{thm:pred-error} is complete.

\subsection{Applications of Theorem~\ref{thm:pred-error}}\label{subsec:pred-appl}

We now present applications of Theorem~\ref{thm:pred-error}. In particular, in Corollary~\ref{cor:appl}, we present robustness functions that satisfy the condition of Theorem~\ref{thm:pred-error}. These are just indicative of what Theorem~\ref{thm:pred-error} can give us, and have the characteristic that the quantity $1/\rho(\eta)$ depends polylogarithmically, logarithmically, and sublogarithmically on the prediction error, respectively. The explicit allocation and payment functions of the corresponding auctions then follow by applying the machinery of Section~\ref{subsec:proof-pred-if}. We remark that for the first $\rho$-function in Corollary~\ref{cor:appl}, which does not depend on $H$, the corresponding auction satisfies the claimed robustness requirement even when it is applied to settings with agent valuations from the interval $[1,\infty)$.

\begin{corollary}\label{cor:appl}
Let $H>1$. For the functions
\begin{align*}
    \rho(\eta)&:=\frac{1}{(\pi/\epsilon+1)\cdot (1+\ln^{1+\epsilon}(\eta))}, \mbox{ with $\epsilon\in (0,1]$},\\
    \rho(\eta)&:=\frac{1}{1+2\ln{(1+\ln{H})}} \cdot \frac{1}{1+\ln{\eta}}, \mbox{ and}\\
    \rho(\eta)&:=\frac{1-\epsilon}{2(1+\ln{H})^{1-\epsilon}}\frac{1}{(1+\ln\eta)^{\epsilon}}, \mbox{ with $\epsilon\in (0,1)$},
\end{align*}
there exist $\rho$-robust DSIC intuitive auctions for bidders with valuations from $[1,H]$ and a prediction for the highest bid. 
\end{corollary}

\begin{proof}
For the first robustness function, let $c_1=\frac{1}{\pi/\epsilon+1}$. Clearly, $\rho(\eta)$ is decreasing in $\eta$. Furthermore, we can verify that $\eta\cdot \rho(\eta)=c_1\cdot \frac{\eta}{1+\ln^{1+\epsilon}{\eta}}$ is non-decreasing in $\eta$. Indeed, the derivative of the fraction is $\frac{1+\ln^{1+\epsilon}{\eta}-(1+\epsilon)\cdot \ln^{\epsilon}{\eta}}{(\pi/\epsilon+1)\cdot \left(1+\ln^{1+\epsilon}{\eta}\right)^2}$. The numerator is of the form $1+y^{1+1/\epsilon}-(1+\epsilon)\cdot y$, which is minimized to $1-\epsilon^\epsilon$ for $y=\epsilon^\epsilon$. Since $\epsilon \in (0,1]$, the derivative is non-negative, and the monotonicity requirements of Theorem~\ref{thm:pred-error} are satisfied.

Also, observe that 
\begin{align}\label{eq:integral}
\int_1^{\infty}{\frac{\ud{z}}{z(1+\ln^{1+\epsilon}{z})}}
    &= \frac{2}{1+\epsilon}\cdot \int_0^{\infty}{\frac{y^{\frac{1-\epsilon}{1+\epsilon}}\cdot \ud{y}}{1+y^2}} = \frac{\pi}{1+\epsilon}\cdot \csc\left(\frac{\pi\epsilon}{1+\epsilon}\right)\leq \frac{\pi}{2\epsilon}.
\end{align}
The first equality follows using substitution $y=\ln^{\frac{1+\epsilon}{2}}{z}$. The second one follows since $\int_0^\infty{\frac{y^{1-\delta}\cdot \ud{y}}{1+y^2}}=\frac{\pi}{2}\cdot \csc\left(\frac{\pi \delta}{2}\right)$ for $\delta\in (0,1]$. For the inequality, observe that $\frac{\pi\epsilon}{1+\epsilon}\in [0,\pi/2]$ and $\sin(z)\geq \frac{2}{\pi}\cdot z$ for $z\in [0,\pi/2]$. Hence, $\csc(z)=\frac{1}{\sin(z)}\leq \frac{\pi}{2z}$. Then, using equation (\ref{eq:integral}), the LHS of the condition in Theorem~\ref{thm:pred-error} becomes
\begin{align*}
    \int_1^{\widehat{u}}{\frac{\rho(z)}{z} \ud{z}}+\int_1^{H/\widehat{u}}{\frac{\rho(z)}{z} \ud{z}}+\rho(H/\widehat{u})&\leq 2c_1\cdot \int_1^{\infty}{\frac{\ud{z}}{z(1+\ln^{1+\epsilon}{z})}}+c_1 \leq c_1\cdot (\pi/\epsilon+1)=1,
\end{align*}
as desired.

We omit the proof for the second and third robustness functions, which are very similar to (actually, simpler than) the above.
\end{proof}

\section{Conclusion}
We have presented an optimal trade-off for the consistency and robustness of DSIC intuitive single-item auctions that use predictions. This result gives us auctions with constant consistency and robustness $\Omega(\ln^{-1}{H})$, where $H$ upper bounds the agent valuations. We have also given a sufficient and necessary condition for DSIC intuitive auctions with a robustness guarantee that is a function of the prediction error. As a corollary, we have obtained auctions with a robustness that depends only on the prediction error $\eta$ (e.g., as  $\Omega(\ln^{-2}{\eta})$ and even better) and not in $H$, or auctions that have a small (logarithmic or sublogarithmic) dependence on $H$ and better dependence on the prediction error. An obvious extension of our work would be to explore trade-offs between consistency and robustness in more general single-parameter mechanism design environments. In addition to the revenue extracted, these can also involve the social welfare. We expect that our variant of Myerson's Lemma will find more applications there.

\bibliography{sample}

\end{document}